\documentclass[11pt]{article}
\usepackage[T1]{fontenc}
\usepackage{amsfonts,amsmath,amsthm,amssymb,dsfont,mathtools}
\usepackage[margin=1in]{geometry}
\usepackage{fullpage}
\usepackage{enumitem}
\usepackage[linesnumbered,boxed,ruled,vlined]{algorithm2e}
\usepackage[colorlinks,citecolor=blue,linkcolor=blue,urlcolor=red,pagebackref]{hyperref}
\usepackage[capitalise]{cleveref}
\usepackage{url}
\usepackage{thmtools, thm-restate}
\usepackage{xcolor}
\usepackage{graphicx}
\usepackage{mathdots}
\usepackage[normalem]{ulem}
\usepackage{xspace}
\xspaceaddexceptions{]\}}
\usepackage{comment}
\usepackage{tabu}
\usepackage{framed}
\usepackage{float,wrapfig}
\usepackage{subfigure}
\usepackage{tikz}
\usepackage[usenames,dvipsnames]{pstricks}

\usepackage[mathlines]{lineno}
\usepackage{etoolbox}
\newcommand*\linenomathpatch[1]{%
  \cspreto{#1}{\linenomath}%
  \cspreto{#1*}{\linenomath}%
  \csappto{end#1}{\endlinenomath}%
  \csappto{end#1*}{\endlinenomath}%
}
\linenomathpatch{equation}
\linenomathpatch{gather}
\linenomathpatch{multline}
\linenomathpatch{align}
\linenomathpatch{alignat}
\linenomathpatch{flalign}

\theoremstyle{plain}
\newtheorem{theorem}{Theorem}[section]
\newtheorem{lemma}[theorem]{Lemma}

\newtheorem{corollary}[theorem]{Corollary}
\newtheorem{claim}[theorem]{Claim}

\theoremstyle{definition}
\newtheorem{definition}[theorem]{Definition}

\setlist[enumerate]{nosep, topsep=1ex}
\setlist[itemize]{nosep, topsep=1ex}
\setlist[description]{nosep}
\allowdisplaybreaks
\def\ShowAuthNotes{1}
\ifnum\ShowAuthNotes=1
\newcommand{\authnote}[2]{\ \\ \textcolor{red}{\parbox{0.9\linewidth}{[{\footnotesize {\bf #1:} { {#2}}}]}}\newline}
\else
\newcommand{\authnote}[2]{}
\fi

\newcommand{\eps}{\varepsilon}

\newcommand{\diam}{\operatorname*{diam}}
\newcommand{\dis}{d}
\newcommand{\MM}{\operatorname*{MM}}
\newcommand{\poly}{\operatorname{\mathrm{poly}}}
\newcommand{\polylog}{\poly\log}

\newcommand{\Z}{\mathbb{Z}}

\title{Improved Additive Approximation Algorithms for APSP}

\author{Ce Jin\thanks{MIT, \texttt{cejin@mit.edu}, supported by the Jane Street Graduate Research Fellowship, NSF grant CCF-2330048, and a Simons Investigator Award.} \and Yael Kirkpatrick\thanks{MIT, \texttt{yaelkirk@mit.edu}, supported by NSF Grant No 2141064.}\and Michał Stawarz\thanks{ETH Zurich, \texttt{m.casi321@gmail.com}.} \and Virginia Vassilevska Williams\thanks{MIT, \texttt{virgi@mit.edu}, supported by NSF Grant CCF-2330048, BSF Grant 2020356 and a Simons Investigator Award.} }

\date{\vspace{-1cm}}

\begin{document}

	\setcounter{page}{0} \clearpage
	\maketitle
	\thispagestyle{empty}
    
	\begin{abstract}
The All-Pairs Shortest Paths (APSP) is a foundational problem in theoretical computer science. Approximating APSP in undirected unweighted graphs has been studied for many years, beginning with the work of Dor, Halperin and Zwick [SICOMP'01]. Many recent works have attempted to improve these original algorithms using the algebraic tools of fast matrix multiplication. We improve on these results for the following problems. 

    For $+2$-approximate APSP, the state-of-the-art algorithm runs in $O(n^{2.259})$ time [D\"urr, IPL 2023; Deng, Kirkpatrick, Rong, Vassilevska Williams, and Zhong, ICALP 2022]. We give an improved algorithm in $O(n^{2.2255})$ time. 
    
    For $+4$ and $+6$-approximate APSP, we achieve time complexities $O(n^{2.1462})$ and $O(n^{2.1026})$ respectively, improving the previous $O(n^{2.155})$ and $O(n^{2.103})$ achieved by [Saha and Ye, SODA 2024].

In contrast to previous works, we do not use the big hammer of bounded-difference $(\min,+)$-product algorithms.
Instead, our algorithms are based on a simple technique that decomposes the input graph into a small number of clusters of constant diameter and a remainder of low degree vertices, which could be of independent interest in the study of shortest paths problems. We then use only standard fast matrix multiplication to obtain our improvements.
    \end{abstract}
	\newpage

\section{Introduction}
\label{sec:intro}

\emph{All-Pairs Shortest Paths (APSP)} is a fundamental problem in computer science: given an edge-weighted graph $G=(V,E)$  with $|V|=n$ vertices, for every pair of vertices $u,v\in V$, compute their distance $d(u,v)$ in the graph $G$. The textbook Floyd--Warshall algorithm  solves APSP in $O(n^3)$ time. The state-of-the-art algorithm by Williams runs in $n^3/2^{\Omega(\sqrt{\log n})}$ time \cite{Williams18}. A central hypothesis in fine-grained complexity asserts that no $O(n^{3-\eps})$-time algorithms can solve APSP in edge-weighted graphs, for any $\eps>0$ (see \cite{finegrainedsurvey}).

For unweighted graphs, better algorithms for APSP are known.
Seidel \cite{Seidel95} gave an algorithm for undirected unweighted APSP in $\widetilde O(n^{\omega})$ time, where $2\le \omega<2.3714$ is the  exponent of fast square matrix multiplication \cite{AlmanDWXXZ25}.
There are also subcubic time APSP algorithms for directed unweighted graphs, and more generally for graphs with small integer weights; see e.g., \cite{AlonGM97,shoshanZwick99,Zwick02}.

In this paper, we focus on \emph{undirected unweighted APSP}. 
There remains a gap between Seidel's $\widetilde O(n^{\omega})$ time and the ideal $\widetilde O(n^2)$ time complexity (which would be nearly optimal).
This gap can be explained by the \emph{Boolean Matrix Multiplication (BMM) hypothesis} from fine-grained complexity, which asserts that multiplying two $n\times n$ matrices over the Boolean semi-ring cannot be solved in $O(n^{\omega-\eps})$ time, for any $\eps>0$\footnote{This hypothesis makes sense if $\omega>2$.}.  Since APSP in undirected unweighted graphs is known to be at least as hard\footnote{It is in fact equivalent to BMM.} as BMM \cite{AingworthCIM99}, this suggests that Seidel's $\widetilde O(n^{\omega})$ time complexity is likely nearly optimal.

Motivated by this situation, many works in the literature have considered approximate APSP in order to bypass this hardness.
In this paper, we focus on \emph{additive approximation}: in an unweighted undirected graph $G=(V,E)$, the \emph{$+C$-APSP} problem asks to compute distance estimates $\tilde d(u,v)$ for every pair of vertices $u,v\in V$, so that $d(u,v)\le \tilde d(u,v)\le d(u,v)+C$ always holds.
It is known that $+1$-APSP is still as hard as BMM
\cite{AingworthCIM99},  so $+2$ is the smallest additive error that allows improvement over Seidel's $\widetilde O(n^\omega)$ time complexity.
Dor, Halperin and Zwick \cite{DHZ00} gave a combinatorial algorithm\footnote{As is typical in the literature, we informally refer to algorithms that do not use fast matrix multiplication as \emph{combinatorial algorithms}.} for $+2$-APSP in $\widetilde O(n^{7/3})$ time, which is faster than Seidel's exact APSP algorithm (for the \emph{current} value of $\omega$). Over two decades later, Deng, Kirkpatrick, Rong, Vassilevska Williams, and Zhong \cite{DengKRWZ22} improved the time complexity to $O(n^{2.2867})$. 
Their key idea was to use Euler tours to design a black-box reduction to the bounded-difference $(\min,+)$-product problem, which is known to have sub-cubic time algorithms \cite{bringmann2019truly,ChiDX022}.\footnote{More specifically, \cite{DengKRWZ22} used fast algorithms that compute the $(\min,+)$-product $C[i,j]=\min_{k}\{A[i,k]+B[k,j]\}$ of two input integer matrices $A,B$, where $A$ is column bounded-difference, i.e., $|A[i,j]-A[i+1,j]|\le O(1)$ for all valid $(i,j)$'s, and $B$ is row bounded-difference, i.e., $|B[i,j]-B[i,j+1]|\le O(1)$ for all valid $(i,j)$'s.}
D\"urr \cite{durr2023improved} further improved the $+2$-APSP time complexity to $O(n^{2.259})$ by developing faster algorithms for rectangular bounded-difference $(\min,+)$-product. 

For $+2k$-APSP ($k\ge 2$),  Dor, Halperin and Zwick \cite{DHZ00} gave faster combinatorial algorithms running in $\widetilde O(n^{2+\frac{1}{3k-1}})$ time. 
Saha and Ye \cite{sahaYeAPSP} improved \cite{DHZ00}'s $+2k$-APSP algorithms by refining the Euler-tour idea of \cite{DengKRWZ22} and the original analysis of \cite{DHZ00}. (See \Cref{tab:results}.)

\subsection{New Results}
In this paper, we give improved algorithms for $+2,+4$, and $+6$-APSP.
In addition to quantitative improvements, our algorithms also simplify previous approaches in the sense that we no longer need to invoke the big hammer of bounded-difference $(\min,+)$-product algorithms. We prove the following two theorems.

\begin{theorem}
\label{thm:mainplus2}
   $+2$-APSP in an $n$-node unweighted undirected graph can be solved by a randomized algorithm in $O(n^{2.22548})$ time.
\end{theorem}

Our algorithm for \Cref{thm:mainplus2} uses rectangular matrix multiplication.
If one only uses square matrix multiplication, the running time of Theorem~\ref{thm:mainplus2} would be $\widetilde{O}(n^{2+\frac{\omega-1}{\omega+3}})$. Thus, the running time will beat Seidel's $\widetilde O(n^\omega)$ time bound as long as $\omega>\sqrt{5}\approx 2.236$.

If $\omega=2$, our running time would be $\widetilde{O}(n^{2.2})$. In contrast, D\"urr's \cite{durr2023improved} algorithm runs in $\widetilde{O}(n^{2+\frac{\omega-1}{2\omega}})$ time in terms of $\omega$ and hence if $\omega=2$, its running time would be $\widetilde{O}(n^{2.25})$ and thus would still be slower than ours even with optimal matrix multiplication bounds.

\begin{restatable}{theorem}{TheMainPlusTwoK}
\label{thm:mainplus2k}
   $+2k$-APSP in an $n$-node unweighted undirected graph can be solved by a deterministic algorithm in $O(n^{2+x/(k+1)})$ time, when $x$ is the solution to $1+x = \omega(1-\frac{k-1}{k+1}x, 1-x, \frac{k}{k+1}x)$. (See the definition of the rectangular matrix multiplication exponent $\omega(\cdot,\cdot,\cdot)$ in \Cref{sec:prelim}.)
\end{restatable}

The result statement of Saha and Ye \cite{sahaYeAPSP} differs from ours only in the equation defining $x$. Theirs has the form $1+\mathbf{2}x=\omega(1-\frac{k-1}{k+1}x, 1-x, \mathbf{1-\frac{k-2}{k+1}}x)$, where the difference from ours appears in boldface.

The result of \Cref{thm:mainplus2k} gives a faster algorithm than the combinatorial $+2k$-approximation of Dor, Halperin and Zwick \cite{DHZ00} for every $k$. Our algorithm uses only  matrix multiplication, without the more complex  algebraic tools used in the work of Saha and Ye \cite{sahaYeAPSP}. 
However, the running time is faster than the best known one of Saha and Ye only for $k=2,3$ (additive +4 and +6 approximation). The results are summarized in \Cref{tab:results}, where all running times (ours and prior work) are computed using the code of \cite{balancer} updated with the newest rectangular matrix multiplication bounds \cite{AlmanDWXXZ25}.

\begin{table}[h]
\centering
\begin{tabular}{|ccccc|}
\hline
\multicolumn{5}{|c|}{$+2k$-Additive Approximation for APSP}                                                   \\ \hline
\multicolumn{1}{|l|}{$2k$} & \multicolumn{1}{c|}{\cite{DHZ00} (combinatorial)} & \multicolumn{1}{c|}{\cite{durr2023improved}} & \multicolumn{1}{c|}{\cite{sahaYeAPSP}}    & This work     \\ \hline
\multicolumn{1}{|c|}{$2$}  & \multicolumn{1}{c|}{$n^{2+1/3} \le  n^{2.334}$}          & \multicolumn{1}{c|}{$n^{2.25899}$}     & \multicolumn{1}{c|}{}              & $n^{2.22548}$ \\
\multicolumn{1}{|c|}{$4$}  & \multicolumn{1}{c|}{$n^{2.2}$}           & \multicolumn{1}{c|}{}                 & \multicolumn{1}{c|}{$n^{2.15492}$} & $n^{2.14613}$ \\
\multicolumn{1}{|c|}{$6$}  & \multicolumn{1}{c|}{$n^{2.125}$}         & \multicolumn{1}{c|}{}                 & \multicolumn{1}{c|}{$n^{2.102926}$} & $n^{2.102595}$ \\
\multicolumn{1}{|c|}{$8$}  & \multicolumn{1}{c|}{$n^{2+1/11}\le n^{2.0910}$}         & \multicolumn{1}{c|}{}                 & \multicolumn{1}{c|}{$n^{2.077270}$} & \textcolor{gray}{$n^{2.079072}$} \\
\hline
\end{tabular}
\caption{Comparison with previous results for $+2k$-APSP.}
\label{tab:results}
\end{table}

\subsection{Further Related Works}

A closely related problem is \emph{multiplicative} approximate APSP. In unweighted undirected graphs, an $+2$-approximation  for APSP automatically yields $2$-multiplicative approximation, so it is an easier problem.  The fastest known algorithm for $2$-multiplicative approximate APSP runs in $O(n^{2.0319})$ time 
\cite{DoryFKNWV24,sahaYeAPSP}, 
improving the previous $\widetilde O(n^{2.25})$-time algorithm by \cite{roditty}.
The main open question in this line of research is to achieve $O(n^{2+o(1)})$ time for $2$-multiplicative approximate APSP.
See \cite{gupta} for very recent progress on this question.

\subsection{Technical Overview}
For a few decades, the $+2k$-approximate APSP algorithm of Dor, Halperin, and Zwick \cite{DHZ00} was the fastest known additive APSP approximation. The core idea of this algorithm is selecting a series of degree thresholds $1=d_0<d_1<d_2<\ldots < d_k$ and sampling hitting sets $S_1,\ldots, S_k$ of size $|S_i|=\tilde{O}(n/d_i)$ that hit the neighborhoods of all vertices of degree $\geq d_i$, for every $i$.

We first compute the distances out of every vertex in the smallest hitting set $S_k$, to compute a correct distance estimate $\tilde{d}(u,v)=d(u,v)$ for any $u\in S_k, v\in V$. Next, we compute the distances out of every vertex in $S_{k-1}$. However, we can't afford to use the full edge set, so instead we only search on the edge set $E_{k-1}$, consisting of edges adjacent to vertices of degree $<d_k$ and an edge from every vertex of degree $\ge d_k$ to a neighbor in $S_k$. Additionally, when running Dijkstra's algorithm from a vertex $u$, we include edges from $u$ to every vertex $v\in S_{k}$ weighted by the current distance estimate computed between them in the previous round. 

Now, for any $u\in S_{k-1},v\in V$, if the shortest path between them contains only vertices of degree $<d_k$, then $\tilde{d}(u,v)=d(u,v)$ since the entire shortest path between them was included in the Dijkstra search. Otherwise, let $w$ be the last vertex on the path from $u$ to $v$ of degree $\ge d_k$. There exists a vertex $s\in S_k$ such that $(s,w)\in E_{k-1}$. Thus, since the path from $w$ to $v$ is also included in $E_{k-1}$, our search out of $u$ considers the path $u\to s \to w \rightsquigarrow v$ of weight $\tilde{d}(u,s) + 1 + d(w,v)$. Since $s\in S_k$ we have that $\tilde{d}(u,s)=d(u,s)\leq d(u,w)+1$ and conclude that $\tilde{d}(u,v)\leq d(u,v) +2$.

We can now iterate this idea. We run Dijkstra's out of every vertex $u\in S_{k-2}$ on the edge set $E_{k-2}$ consisting of edges adjacent to vertices of degree $<d_{k-1}$ and edges connecting vertices of degree $\ge d_{k-1}$ to a neighbor of their in $S_{k-1}$, in addition to an edge out of $u$ to every vertex $v$ weighted by the current best distance estimate between the pair. The same argument shows that all distance estimates computed out of $S_{k-2}$ will be within +4 of the true distance. Repeating this argument grows the additive error by +2 with every iteration, resulting in a $+2k$ error for distances computed out of $S_0=V$.

Balancing the degree thresholds gives a running time of $\widetilde{O}(n^{2-1/(k+1)}m^{1/(k+1)})$, which is the best known running time for sparse approximate APSP. However, for dense graphs Dor, Halperin and Zwick introduce additional edges to the graph search that improve the additive approximation to $<2k$, while keeping a running time of $\widetilde{O}(n^{2+1/(k+1)})$.

Recent work has sought to improve the dense approximate APSP algorithm using fast matrix multiplication. Deng et al. \cite{DengKRWZ22} noted that for a $+2$-approximation, we can avoid running Dijkstra's out of $V=S_{k-1}$ by considering two cases. For pairs of vertices such that the shortest path between them contains only vertices of degree $<d_1$, run the sparse approximate APSP algorithm of Dor, Halperin and Zwick. For the remaining pairs, they show how to compute the distances out of $S_1$ in $\widetilde{O}(n^2)$ time and then compute $\tilde{d}(u,v)=\min_{s\in S_1} d(u,s) + d(s,v)\leq d(u,x)+2$ using the $(\min,+)$ product.

In general, the $(\min,+)$-product (likely) cannot be solved polynomially faster than brute force, as it is equivalent to the APSP problem. However, the authors of \cite{DengKRWZ22} showed that one can sort the vertices of the graph in a way (based on the Euler tour of a spanning tree) that the resulting matrices are \emph{column/row-bounded-difference} matrices. For such matrices, there exists a subcubic time algorithm for computing their $(\min,+)$-product.
The following work of D\"urr \cite{durr2023improved} used faster rectangular $(\min,+)$-product to speed up the +2 additive approximation.

Saha and Ye \cite{sahaYeAPSP} used this same idea to improve the general $+2k$-approximation. Their algorithm replaces the first two stages of Dijkstra's searches with a call to sparse approximate APSP and a $(\min,+)$-product of the matrices representing the distances between $S_{k-1}, S_{k}$ and $S_k,S_{k-2}$.

In this paper, we introduce a new way to speed up the $(\min,+)$-product computation. Instead of using the subcubic, but still considerably slow and complicated, bounded-difference $(\min,+)$-product, we reduce the problem to a $(\min,+)$-product with entries bounded by a constant. Due to a standard reduction (e.g. \cite{shoshanZwick99}), such a product can be computed in fast matrix multiplication time.

To achieve this, we introduce a new graph decomposition technique which decomposes the graph into a small number of clusters of constant diameter and a remainder of low degree vertices. We then compute the (min,+)-product on each cluster independently. As all vertices in the cluster are within constant distance of each other, we can shift the values of the corresponding distance matrix such that the resulting matrix is bounded by a constant. This allows us to compute distances between pairs of vertices in the clusters. To extend our approximation to the vertices in the remainder, we run a sparse graph search out of every vertex on just the low degree edges adjacent to vertices in the remainder.

\subsection{Organization}
In \Cref{sec:prelim} we prove the general decomposition lemma which we will later use in all of our algorithms. In \Cref{sec:faster+2} we prove \Cref{thm:mainplus2} in two stages, beginning with a slower, simpler, `warm up' algorithm that already beats the current state of the art algorithm for $+2$-APSP. Finally, in \Cref{sec:faster+2k} we prove \Cref{thm:mainplus2k} by extending the simpler algorithm of \Cref{sec:faster+2} to a general $+2k$-approximation.
We conclude with open questions in \Cref{sec:open}.

\section{Preliminaries}
\label{sec:prelim}

Let $[n]=\{1,2,\dots,n\}$.
Let $G=(V,E)$  be an undirected graph and $U\subseteq V$ be a vertex subset. Let $G[U]$ denote the subgraph of $G$ induced by $U$. Let $\deg_{G}(u)$ denote the degree of vertex $u$ in graph $G$.
Let $\dis_G(u,v)$ denote the distance between vertices $u$ and $v$ in graph $G$, and let $P_G(u,v)$ denote the shortest path from $u$ to $v$, including the endpoints $u,v$ (if more than one shortest path exists, for convenience we pick one that maximizes $\max_{x\in P_G(u,v)}\deg_G(x)$).
Let $\diam_G(U)$  denote the (weak) diameter of the vertex subset $U$, defined as $\diam_G(U) = \max_{u,v\in U} \dis_G(u,v)$.
We omit the subscript $G$ and simply write $\deg(u)$, $\diam(U)$, $\dis(u,v)$ and $P(u,v)$ when the underlying graph $G$ is clear from the context.

For a path $P$, let $|P|$ denote the length of $P$ (in unweighted graphs, $|P|$ equals the number of edges in $P$).

We call a distance estimate $\tilde{d}(u,v)$ an additive $+C$-approximation if for every pair $u,v\in V$ the estimate satisfies $d(u,v)\leq \tilde{d}(u,v)\leq d(u,v) + C$.

In our algorithms we use the following (combinatorial) algorithm for sparse additive approximate APSP by Dor, Halperin, and Zwick \cite{DHZ00}.
\begin{lemma}[\cite{DHZ00}]
\label{lem:sparse}
   $+2k$-Approximate APSP on an $n$-node $m$-edge unweighted undirected graph can be solved in $\widetilde O(n^{2-1/(k+1)}m^{1/(k+1)})$ time.
\end{lemma}

We will use this lemma on paths containing vertices of bounded degree, in which case we obtain the following corollary by considering the $O(nd)$ edges adjacent to vertices of degree $\leq d$.

\begin{corollary}
    \label{cor:sparse}
    $+2k$-Approximate APSP between pairs of points such that a shortest path between them uses only vertices of degree $\leq d$ can be solved in $\widetilde{O}(n^2d^{1/(k+1)})$ time.
\end{corollary}

A standard tool used in these algorithms is the \emph{hitting set}, as defined in the following lemma.

\begin{lemma}[Hitting set, e.g., {\cite[Theorem 2.7]{AingworthCIM99}}]
\label{lem:dethittingset}
Given an $n$-node undirected graph $G=(V,E)$ and a degree threshold $1\le d\le n $, one can deterministically  construct in $O(n^2)$ time a hitting set $S\subseteq V$ of size $O(\frac{n\log n}{d})$ such that every node $u\in V$ of degree at least $d$ in $G$ is adjacent to some $s\in S$.
\end{lemma}

Let $\MM(n_1,n_2,n_3)$ denote the time complexity of multiplying an $n_1\times n_2$ matrix by an $n_2\times n_3$ matrix. We denote by $\omega(\gamma_1, \gamma_2,\gamma_3)$ the exponent of $\MM(n^{\gamma_1}, n^{\gamma_2},n^{\gamma_3})$, i.e. the minimum value $c$ such that the product of an $n^{\gamma_1}\times n^{\gamma_2}$ matrix by an $n^{\gamma_2}\times n^{\gamma_3}$ matrix can be computed in $O(n^{c+\varepsilon})$ time for any $\varepsilon>0$. 

A common matrix product used in shortest path computation is the $(\min,+)$-product, defined as follows.

\begin{definition}[$(\min,+)$-matrix product]
    The $(\min,+)$-product of two matrices $A,B$ is defined as $C=A\star B$ where $C[i,j]\coloneqq \min_{k} A[i,k]+B[k,j]$.
\end{definition}

When the entries of the matrices are bounded by an integer $L$, a standard method which encodes the entries as polynomials of degree $O(L)$ allows to compute their $(\min,+)$-product in fast matrix multiplication time (see e.g. \cite{shoshanZwick99}).

\begin{lemma}[\cite{shoshanZwick99}]\label{lem:fastminplus}
    Given an $n_1\times n_2$ matrix $A$ and an $n_2\times n_3$ matrix $B$ such that entries of both matrices are in $\{0,1,\ldots, L,\infty\}$, computing $C=A\star B$ can be done in time $\widetilde{O}(L\cdot \MM(n_1, n_2,n_3))$.
\end{lemma}

\subsection{A Decomposition Lemma}
\label{sec:decompositionlemma}

Next, we prove the following decomposition lemma, which is a key component in our new approximation algorithms. The lemma shows that for any threshold $d$ we can decompose our graph into disjoint clusters of size greater than $d$ with constant diameter. The remaining vertices that are not assigned to clusters will all have degree smaller than $d$. While the decomposition itself is quite simple, the bounded diameter of the resulting clusters making up the graph is crucial in allowing for faster computation of their approximate shortest paths.

\begin{lemma}
\label{lem:decomp}
  Let $1\le d\le n$. Given an $m$-edge undirected unweighted graph $G=(V,E)$, in $O(m)$ time we can deterministically decompose $V$ into a disjoint union  $R \cup \bigcup_{i=1}^h H_i$  such that:
   \begin{itemize}
   \item $|H_i|> d$ for all $i\in [h]$.
       \item $\diam(H_i)\coloneqq \max_{u,v\in H_i} \dis_{G}(u,v)\le 4$ for all $i\in [h]$.
       \item $\deg_{G}(u) <d$ for all $u\in R$.
   \end{itemize}
\end{lemma}
\begin{proof}
We begin constructing the clusters iteratively. Initialize $U=V$. While there exists $u\in U$ such that $\deg_{G[U]}(u)\ge d$, create a new cluster $H_i'$ containing $u$ and all of $u$'s neighbors in $U$. Remove the vertices of $H_i'$ from $U$ and iterate until no vertex in $G[U]$ has degree $\geq d$ and we have created $h$ clusters $H_1',\ldots,H_h'$. Note that by taking $H_i'$ to be the neighborhood of a high degree vertex we guarantee that $|H_i'|>d$ and $\diam(H_i') \leq 2$.

Next, beginning with $i=1$ and iterating over the clusters, we define $H_i''$ to be the vertices in $U$ that are adjacent to a vertex in $H_i'$, and remove $H_i''$ from $U$. We return $H_i\coloneqq H_i' \cup H_i''$ and set $R$ to be the remaining vertices in $U$ at the end of this process.

When we can no longer create new clusters at the first stage, we have that every vertex $u\in U$ has $<d$ neighbors in $U$. 
Thus, for every $u\in U$, if $\deg_G(u) \geq d$, it must have a neighbor in $H_i'$ for some $i\in [h]$ and so this vertex will have been added to $H_i''$ and removed from $U$.  We conclude that all vertices $u\in  R$ have degree $\deg_G(u)<d$. Finally, since $\diam(H_i') \leq 2$ and all vertices in $H_i''$ are adjacent to $H_i'$ we have that $\diam(H_i) \leq 4$ for every $i\in [h]$. Since $|H_i'|>d$ we also have that $|H_i|>d$.
\end{proof}

\subsection{A Min-Plus Lemma}
Lastly, we prove the following lemma about computing the $(\min,+)$ product of a matrix representing the distances from a set of low diameter.
\begin{lemma}
\label{lem:minplus}
Let integer parameter $L\ge 1$.
   Given input matrices $A \in \Z^{n_1\times n_2}$ and $B\in \Z^{n_2\times n_3}$ such that $|B[k,j]-B[k,j']|\le L$ for all $k\in [n_2]$ and all $j,j'\in [n_3]$,  we can compute the $(\min,+)$-product of $A$ and $B$ in $\widetilde O(L\cdot \MM(n_1,n_2,n_3))$ time.
\end{lemma}
\begin{proof}
For $k\in [n_2]$, let $\Delta_k \coloneqq \min_{j\in [n_3]}B[k,j]$, and define matrix $B'\in \Z^{n_2\times n_3}$ by $B'[k,j]\coloneqq B[k,j]-\Delta_k$. 
Then, all entries of $B'$ are in $\{0,1,\dots,L\}$.
Define matrix $A'\in \Z^{n_1\times n_2}$ by $A'[i,k]\coloneqq A[i,k]+\Delta_k$. Then the $(\min,+)$-product of $A$ and $B$ equals the $(\min,+)$-product of $A'$ and $B'$, so it suffices to compute the latter (and we denote the answer by $C\in \Z^{n_1\times n_3}$).
   
For $i\in [n_1]$,  let $m_i \coloneqq \min_{k\in [n_2]} A'[i,k]$. Then, since $0\le B'[k,j]\le L$, we have $C[i,j] = \min_{k\in [n_2]}\{A'[i,k]+B'[k,j]\} \in [m_i, m_i+L]$ for all $i,j$. 
Thus, if $A'[i,k_0]>m_i+L$, then $A'[i,k_0]$ is useless since $A'[i,k_0]+B'[k_0,j]>C[i,j]$ for all $j$, so we can replace the entry $A'[i,k_0]$ by $+\infty$ without changing the $(\min,+)$-product of $A'$ and $B'$.
Then, every entry $A'[i,k]$ is either $+\infty$ or in $[m_i,m_i+L]$. Define the matrix $A''\in \Z^{n_1\times n_2}$ by $A''[i,k]\coloneqq A'[i,k]-m_i \in \{0,1,\dots,L,+\infty\}$.
We compute the $(\min,+)$-product $C''$ between $A'' \in \{0,1,\dots,L,+\infty\}^{n_1\times n_2}$ and $B'\in \{0,1,\dots,L\}^{n_2\times n_3}$ in $\widetilde O(L\cdot \MM(n_1,n_2,n_3))$ time using \Cref{lem:fastminplus}. 

Finally, return the answer matrix $C[i,j] = C''[i,j]+m_i$.
\end{proof}

\section{Faster +2-Approximate APSP}\label{sec:faster+2}

\subsection{A Warm-Up Algorithm}\label{sec:warmup+2}

Now we describe our algorithms for computing +2-Approximate APSP based on the decomposition given in \Cref{lem:decomp}. In this section we present a warm-up algorithm (\Cref{cor:warmup}) which already improves over the state-of-the-art algorithms \cite{DengKRWZ22,durr2023improved}.  In the next section we will give further improvement using more technical ideas.

We use a standard argument (which was also used in \cite{DengKRWZ22,durr2023improved}) that assumes the nodes on the (true) shortest paths under consideration have maximum degree $\Theta(D)$. More specifically, we will prove the following main lemma:

\begin{lemma}
\label{lem:new}
 Let $G=(V,E)$ be an $n$-vertex undirected unweighted graph, and parameter $1\le D\le n$.
    Let $d_D(u,v)$ denote the minimum length of any (not necessarily simple) path $P$ from $u$ to $v$ such that $\max_{x\in P} \deg(x) \in [D,2D]$.
 We can compute distance estimates $\tilde d(u,v)$ such that $d(u,v)\le \tilde d(u,v) \le d_D(u,v)+2$ holds for all $(u,v)\in V^2$, by a deterministic algorithm with time complexity 
   \[ \widetilde O\left (\min_{1\le d<D} \left \{n^2 d + \frac{n}{d} \MM\left (n,\frac{n}{D},d\right )\right \}\right ).\]
\end{lemma}

We analyze the overall time complexity for $+2$-Approximate APSP obtained from 
\Cref{lem:new}:

\begin{corollary}[Warm-Up +2-APSP]
\label{cor:warmup}
   +2-Approximate APSP has a deterministic algorithm in $O(n^{2.2548})$ time.
\end{corollary}
\begin{proof}
We enumerate $1\le D\le n$ that are powers of two, and compute distance estimates $\tilde d(u,v) \in [d(u,v),d(u,v)+2]$ for pairs $(u,v)$ satisfying $\max_{x\in P(u,v)}\deg(x) \in [D,2D]$ (note that $d_D(u,v)=d(u,v)$ holds for such $u,v$). Finally we combine the answers across all $D$. Then, for given $D$, we can without loss of generality assume the input graph has maximum degree at most $2D$, and hence number of edges at most $Dn$.
We run either \Cref{lem:new} or  \Cref{cor:sparse} (whichever is faster) to compute the distance estimates. The time complexity of \Cref{cor:sparse} is $\widetilde O(n^2 \sqrt{D}) $.  The overall time complexity for $+2$-APSP is thus
   \[ \widetilde O\left (\max_{1\le D\le n} \min \left \{n^2\sqrt{D}, \min_{1\le d<D} \left \{n^2 d + \frac{n}{d} \MM\left (n,\frac{n}{D},d\right )\right \} \right \}\right ).\]
Using current fastest algorithms for rectangular matrix multiplication \cite{AlmanDWXXZ25}, the time complexity can be bounded by $n^{2.2548}$ \cite{balancer}.

If we only use square matrix multiplication, then in the above expression it is optimal to pick $d = (n/D)^{1/(4-\omega)} \le n/D$, so the time complexity becomes \[\widetilde O(\max_{D} \min\{n^2\sqrt{D}, n^2 d +  \tfrac{n}{d} \cdot (n/d) (\tfrac{n}{D}/d) d^\omega\}) = \widetilde O(\max_{D}\min\{n^2\sqrt{D}, n^2 (\tfrac{n}{D})^{1/(4-\omega)}\}) = \widetilde O( n^{2+\frac{1}{6-\omega}} ).\]

\end{proof}


Now we describe our algorithm for given degree parameter $D$:
\begin{proof}[Proof of \Cref{lem:new}]
    Without loss of generality, we can assume the input graph $G=(V,E)$ has maximum degree at most $2D$, and hence number of edges $m\le Dn$.
    
    Deterministically construct a hitting set (\cref{lem:dethittingset}) $S \subset V$ of size $O(\frac{n\log n}{D})$ such that every $x\in V$ with $\deg(x)\ge D$ is adjacent to some vertex $s_x\in S$. Thus, if $d_D(u,v)$ is realized by the path $P$ which connects $u,v$ and contains a vertex $x$ of degree $\deg(x)\ge D$, then 
    \begin{equation}
    \label{eqn:temp-plustwo}
    d(u,v)\le \min_{s\in S}\{d(u,s)+d(s,v)\}\le d(u,s_x)+d(s_x,v) \le  d(u,x)+1 + d(x,v)+1 = d_D(u,v)+2.
    \end{equation}
    We run BFS from every $s\in S$ on $G$ to compute $d(s,u)$ for all $(s,u)\in S\times V$, in total time $O(|S| m)   \le \widetilde O(n/D)\cdot Dn =  \widetilde O (n^2)$.

Let $1\le d<D$ be a tunable parameter.
   Run \Cref{lem:decomp} with parameter $d$, and obtain the vertex partition $V =R \cup \bigcup_{i=1}^h H_i$, where each cluster $H_i$ has size $|H_i|\ge d$ and (weak) diameter $\diam(H_i)\le 4$.
For convenience,  we assume without loss of generality that $|H_i|\le 2d$ for all $i$ (by possibly breaking larger clusters into smaller ones, which does not increase their weak diameter), and we still have $h = O(n/d)$ clusters in total.

    For each cluster $H_i$, we compute $\tilde d[V,H_i]$ as the $(\min,+)$ product between the distance matrices $d[V, S]$ and $d[S, H_i]$ using \Cref{lem:minplus}. 
    Note that the parameter $L$ in \Cref{lem:minplus} can be bounded using the triangle inequality by 
    $L = \max_{s\in S,v,v'\in H_i}|d(s, v)- d(s,v')|\le \max_{s\in S,v,v'\in H_i}d(v,v') = \diam(H_i)\le O(1)$. Thus, the  total time for invoking \Cref{lem:minplus} over all $h=O(n/d)$ clusters is
   \[ O\left (\frac{n}{d} \right )\cdot \widetilde O\big ( L \cdot \MM(|V|,|S|,|H_i|)\big )= \widetilde O\left ( \frac{n}{d}\MM(n,n/D,d) \right ) .\]
By \Cref{eqn:temp-plustwo}, $\tilde d(v,h)$ (i.e., the $(v,h)$-th entry in the $(\min,+)$-product between $d[V, S]$ and $d[S, H_i]$) is a $+2$-approximation of $d_D(v,h)$, so we have already computed the desired answers for pairs $(v,h)\in V\times \bigcup_{i=1}^h H_i$.
 It remains to compute answers for the pairs $(v,r)\in V\times R$.

  We enumerate each $v\in V$, and compute answers for the pairs $v\times R$ using the following lemma:

\begin{lemma}
   \label{lem:extendtoall}
For $v\in V$, suppose for all $w\in V\setminus R$ we know distance estimates $\tilde d(v,w)$ such that
$ d(v,w) \le \tilde d(v,w) \le d_D(v,w)+2$.
Then, in $\widetilde O(nd)$ time, we can compute distance estimates $\tilde d(v,w)$ such that $ d(v,w) \le \tilde d(v,w) \le d_D(v,w)+2$, for all $w\in V$.
\end{lemma}
\begin{proof}
  Define a (weighted) auxiliary graph $G_v$ on the same vertex set $V$ as follows: 
 \begin{itemize}
     \item For every $r\in R$, include all the neighboring edges of $r$ into $G_v$. 
     Since $\deg_G(r)< d$ by \Cref{lem:decomp}, we have added $\sum_{r\in R}\deg_G(r) \le |R|d\le nd$ edges.
     \item  For every $h\in \bigcup_{i=1}^h H_i$, add an edge into $G_v$ between $v$ and $h$ with edge weight $\tilde d(v,h)$. 
      This step adds only $\leq n$ edges to $G_v$.
 \end{itemize} 
 Then, use Dijkstra's algorithm to compute the distances from $v$ to all vertices on $G_v$ in $\widetilde O(|E(G_v)|) = \widetilde O(nd)$ time, and return these results as our distance estimates $\tilde d(v,r)$ for all $r\in R$. 
   By construction of $G_v$, it is clear that $d_{G_v}(v,r)$ does not underestimate the true distance $d_G(v,r)$, so it remains to prove that they achieve $+2$-approximation. 

Suppose $d_D(v,r)$ is realized by the path $P$ from $v$ to $r$ with $\max_{x\in P}\deg(x) \ge D$.
 If all vertices on $P$ are contained in $R$, then by construction of $G_v$, all edges on the path $P(v,r)$ are included in $G_v$, so $d_{G_v}(v,r) \le d_D(v,r) < d_D(v,r)+2$ as claimed.
 It remains to consider the case where $P$ is not fully contained in $R$.
  Let $h$ be the last vertex on $P$ such that $h\notin R$ (we know $h$ exists and $h\neq r\in R$). Let $r_1,r_2,\dots,r$ be the vertices after $h$ on the path $P$. Then, $r_1,r_2,\dots,r\in R$ by definition of $h$.
  By construction of $G_v$, this means the edges on the suffix $P_{hr} = h\to r_1\to r_2\to  \cdots \to r$ of the path $P$ are included in $G_v$.
  Since $G_v$ also contains edge $(v,h)$ of weight $\tilde d(v,h)$, we obtain 
  \[d_{G_v}(v,r) \le \tilde d(v,h) + |P_{hr}| \le  (d_D(v,h)+2) +|P_{hr}| \le d_D(v,r)+2\]
  as desired, where the second inequality follows from the input assumption and $h\in V\setminus R$, and the  last inequality is justified as follows:
  Recall $\max_{x\in P}\deg(x) \in [D,2D]$, whereas the vertices  after $h$ on  path $P$, namely $r_1,r_2,\dots,r$, all belong to $R$ and hence have degree $<d<D$. Therefore, the prefix $v\leadsto h$ of the path $P$ contains a node of degree in $[D,2D]$, which implies $d_D(v,h) \le d_D(v,r) - |P_{hr}|$ as required. 
\end{proof}

 The total time of applying 
   the above \Cref{lem:extendtoall} for all $v\in V$ is $\widetilde O(n^2 d)$. This finishes the proof of \Cref{lem:new}.
\end{proof}

\subsection{A Faster Algorithm}

In the previous algorithm, we separately computed the $(\min,+)$-product between $d[V,S]$ and $d[S,H_i]$ for every cluster $H_i$. In this section, we obtain a speed-up by considering all clusters together. More specifically, the first step of our improved algorithm is to efficiently compute the $(\min,+)$-product between $d[\bigcup_{i=1}^h H_i,S]$ and $d[S,\bigcup_{i=1}^h H_i]$. This step uses the technique of false positives mod prime, which was also a crucial ingredient in the state-of-the-art bounded-difference $(\min,+)$-product algorithm \cite{ChiDX022}. This step can be formalized as the following technical lemma:

\begin{lemma}
   \label{lem:minplusgroups}
   Let integer parameter $L\ge 1$.
   Let $A\in \Z^{hd\times s}, B\in \Z^{s\times hd}$ be input matrices with entries in $[-U,U]$. 
   Partition the indices $i\in \{1,2,\dots,hd\}$ into $h$ contiguous groups each of size $d$ (so that $i$ belongs to the $\lceil i/d\rceil$-th group).

   Suppose $|A[i,k]-A[i',k]|\le L$ holds for all $k\in [s]$ and all $i,i'$ in the same group, and 
    $|B[k,j]-B[k,j']|\le L$ holds for all $k\in [s]$ and all $j,j'$ in the same group.

Then, for any parameter $q\ge 1$, we can compute the $(\min,+)$-product of $A$ and $B$ by a randomized algorithm in time complexity
      \[\widetilde O\left (h^2 s + qL \cdot  \MM(hd,s,hd) + h^2 L\cdot  \MM(d,s/q,d)\right )\cdot \polylog(U). \]
   \end{lemma}
   \begin{proof}
      Throughout, we use the shorthand $\hat i = \lceil i/d\rceil$ and $\hat j = \lceil  j/d\rceil$ to denote the groups that contain indices $i$ and $j$ respectively.

 Let $C\in \Z^{hd\times hd}$ denote the $(\min,+)$-product of $A$ and $B$ which we want to compute. 
  We first efficiently compute an $O(L)$-additive approximation of $C$ as follows: 
      Define matrix $A'\in \Z^{h\times s}$ by \[A'[\hat i,k] \coloneqq \left\lfloor \tfrac{1}{L} \min_{(\hat i-1)d<i \le \hat id}  A[i,k]\right\rfloor,\]
      and define matrix $A'' \in \Z^{hd\times s}$ by 
      \begin{equation}
         \label{eqn:defnapp}
A''[i,k]\coloneqq A[i,k] - L\cdot A'[  \hat i , k].
      \end{equation}
Then, observe that the input assumption on $A$ implies 
\begin{equation}
   \label{eqn:approxA}
0\le A''[i,k] < 2L
\end{equation}
 for all $i\in [hd]$ and $k\in [s]$. 
We define matrices $B'\in \Z^{s\times h}$ and $B''\in \Z^{s\times hd}$ analogously with
\begin{equation}
   \label{eqn:defnbpp}
   B''[k,j]\coloneqq B[k,j] - L \cdot B'[k,\hat j],
\end{equation}
and 
\begin{equation}
   \label{eqn:approxB}
 0\le B''[k,j]<2L
\end{equation}
for all $k\in [s],j\in [hd]$.
We compute $C'\in \Z^{h\times h}$ as the $(\min,+)$-product of $A',B'$ by brute force in time $O(h^2s)$.
From \Cref{eqn:approxA,eqn:approxB} we can observe that $C'$ provides an $O(L)$-additive approximation of the true $(\min,+)$-product $C$ in the sense that 
\begin{equation}
   \label{eqn:approxl}
  0 \le C[i,j] - L\cdot C'[\hat i,\hat j] < 4L
\end{equation}
 holds for all $i,j\in [hd]$.
Thus, if $k$ is a witness for $C[i,j]$ (i.e., $A[i,k]+B[k,j] = C[i,j]$), then 
\begin{align}
\left \lvert A'[\hat i,k]+B'[k,\hat j]-C'[\hat i,\hat j]\right \rvert & = \frac{1}{L}\left \lvert  - A''[i,k]-B''[k,j] + A[i,k] + B[k,j]   -L\cdot C'[\hat i,\hat j]\right \rvert \nonumber \\
& = \frac{1}{L}\left \lvert  - A''[i,k]-B''[k,j] + C[i,j] -L\cdot C'[\hat i,\hat j]\right \rvert \nonumber \\
& < 4,
   \label{eqn:roundadditiveerr}
\end{align}
where the last step follows from \Cref{eqn:approxA,eqn:approxB,eqn:approxl}.
 
Pick a random prime $p \in \Theta(q \log U)$.
Define an $hd\times s$ matrix $\tilde A$ where each entry is a bivariate monomial defined as
\[ \tilde A[i,k]\coloneqq  x^{A'[ \hat i , k] \bmod p}y^{A''[i,k]},\]
which has $x$-degree less than $p  = O(q\log U)$ and $y$-degree less than $ 2L$ (by \Cref{eqn:approxA}). 
Define $s\times hd$ matrix $\tilde B$ analogously.
Compute $hd\times hd$ matrix $\tilde C$ as the product of $\tilde A$ and $\tilde B$ via Fast Matrix Multiplication and FFT (e.g., \cite{shoshanZwick99}) in
$\widetilde O( \MM(hd,s,hd)\cdot (q \log U)\cdot (2L))$ time, so
\[ \tilde C[i,j] = \sum_{k\in [s]} x^{(A'[ \hat i , k] + B'[k, \hat j ]) \bmod p} y^{A''[i,k] + B''[k,j]}.\]

Recall every witness $k$ for $C[i,j]$ satisfies \Cref{eqn:roundadditiveerr}. In particular, $k$ satisfies the following mod-$p$ version of \Cref{eqn:roundadditiveerr}: 
\begin{equation}
 A'[ \hat i , k] + B'[k, \hat j ] \in C'[ \hat i, \hat j] + \{-3,\dots,3\} \pmod{p}. 
 \label{eqn:modpeq}
\end{equation}
Hence $k$ contributes a term in the polynomial $\tilde C[i,j]$ with $x$-degree in $C'[ \hat i,  \hat j] +\{-3,\dots,3\} \pmod{p}$. 
We say a triple $(\hat i,k,\hat j) \in [h]\times [s]\times [h]$ is a \emph{false positive} if it does not satisfy \Cref{eqn:roundadditiveerr},  but satisfies \Cref{eqn:modpeq}.
 Then,  for every $(i,j)\in [hd]\times [hd]$, compute the polynomial
\begin{equation}
   \label{eqn:falsepositivesubtracted}
 \tilde C[i,j] -  \sum_{k: (\hat i,k,\hat j)\text{ is a false positive}}x^{(A'[ \hat i , k] + B'[k, \hat j ]) \bmod p} y^{A''[i,k] + B''[k,j]}. 
\end{equation}
We describe how to compute these polynomials later.

Consider each $(i,j)\in [hd]\times [hd]$. We enumerate every non-zero $x^{c'}y^{c''}$ term in the polynomial \Cref{eqn:falsepositivesubtracted}
that satisfies $c'\in C'[\hat i,\hat j]+\{-3,\dots,3\}\pmod{p}$. Then, this term must originate from some
$x^{(A'[ \hat i , k] + B'[k, \hat j ]) \bmod p} y^{A''[i,k] + B''[k,j]}$ where  $A'[ \hat i , k] + B'[k, \hat j]$ equals the unique integer in $C'[\hat i,\hat j]+\{-3,\dots,3\}$ that is congruent to $c'$ modulo $p$ (since the contribution from false positives has been removed), and
$A''[i,k]+B''[k,j]=c''$.
Then,  the corresponding value of $A[i,k]+B[k,j]$ can be recovered via \Cref{eqn:defnapp,eqn:defnbpp} as
\[ A[i,k]+B[k,j] = L\cdot (A'[\hat i,k] +B'[k,\hat j])+(A''[ i,k] +B''[k, j]),\]
and we use it to update the candidate answer for $C[i,j]$. This algorithm correctly computes all $C[i,j]$ since all potential witnesses $k$ for each $C[i,j]$ have been considered.

Now it remains to analyze the total time complexity for computing the polynomial in \Cref{eqn:falsepositivesubtracted} for all $(i,j)\in [hd]\times [hd]$; since the first term $\tilde C[i,j]$ is already computed, we focus on computing the second term (summation over $k$) in \Cref{eqn:falsepositivesubtracted}.
For $\hat i,\hat j\in [h], r\in \{-3,\dots,3\}$, define set 
(which can be computed in $O(h^2 s)$ time by brute force enumeration)
\[F_{\hat i,\hat j,r}\coloneqq \Big \{k\in [s]: (\hat i,k,\hat j)\text{ is a false positive and }A'[\hat i,k]+B'[k,\hat j]\equiv C'[\hat i,\hat j]+ r \pmod{p}\Big \}.\]
Fix $(\hat i,\hat j)\in [h]\times [h]$. The second term in \Cref{eqn:falsepositivesubtracted} for any $(i,j)\in ((\hat i-1)d,\hat id]\times  ((\hat j-1)d,\hat jd]$ can be written as
\[ \sum_{r\in \{-3,\dots,3\}}x^{(C'[\hat i,\hat j]+r)\bmod p}\sum_{k\in F_{\hat i,\hat j,r}} y^{A''[i,k]}\cdot y^{B''[k,j]}.\]
We can compute this for all $(i,j)\in ((\hat i-1)d,\hat id]\times  ((\hat j-1)d,\hat jd]$ using matrix multiplications of dimension 
$d\times |F_{\hat i,\hat j,r}| \times d$ for $r\in \{-3,\dots,3\}$ where each entry is a degree-$O(L)$ univariate polynomial in $y$, in time $\sum_{r\in \{-3,\dots,3\}}\widetilde O(\MM(d,|F_{\hat i,\hat j,r}|,d) \cdot L)$.
Note that $k\in F_{\hat i,\hat j,r}$ holds only if $p$ is a prime factor of the non-zero integer
$A'[\hat i,k]+B'[k,\hat j]- C'[\hat i,\hat j]- r$, which happens with probability (over random prime $p\in \Theta(q \log U)$) at most $O(1/q)$ by the prime number theorem.  Hence, by linearity of expectation and Markov's inequality, for fixed $(\hat i,\hat j)\in [h]^2$, we have $\sum_{r\in \{-3,\dots,3\}} |F_{\hat i,\hat j,r}| = O(s/q)$ with $\ge 0.99$ success probability. In this successful case, the time for computing 
\Cref{eqn:falsepositivesubtracted} (which then allows to compute $C[i,j]$) for all $(i,j)\in ((\hat i-1)d,\hat id]\times  ((\hat j-1)d,\hat jd]$ becomes 
$\widetilde O(\MM(d,s/q,d) \cdot L)$.
The total time over all $(\hat i,\hat j)\in [h]^2$ is $\widetilde O(h^2 \MM(d,s/q,d) \cdot L)$.
We repeat the whole process $O(\log h)$ times with independent random primes $p$, so that with high probability every $(\hat i,\hat j)\in [h]^2$ is successful at least once.


The total time is
      $\widetilde O(h^2 s +  \MM(hd,s,hd)\cdot (q \log U)\cdot (2L) + h^2   \MM(d,s/q,d)\cdot L)$ as claimed.
   \end{proof}
   
   Now we use \Cref{lem:minplusgroups} to prove the following lemma:
   
\begin{lemma}
\label{lem:newnew}
 Let $G=(V,E)$ be an $n$-vertex undirected unweighted graph, and parameter $1\le D\le n$.  
Let $d_D(u,v)$ denote the minimum length of any (not necessarily simple) path $P$ from $u$ to $v$ such that $\max_{x\in P} \deg(x) \in [D,2D]$.
 We can compute distance estimates $\tilde d(u,v)$ such that $d(u,v) \le \tilde d(u,v) \le d_D(u,v) + 2$ for all $(u,v)\in V^2$, by a randomized algorithm with time complexity 
   \[ \widetilde O\left (\min_{1\le d<D, q\ge 1} \left \{n^2 d + \left (\frac{n}{d}\right )^2\cdot \frac{n}{D} + q\cdot \MM \left (n,\frac{n}{D},n\right )  + \left (\frac{n}{d}\right )^2\cdot \MM\left (d,\frac{n}{Dq},d\right )\right \}\right ). \] 
\end{lemma}

\begin{proof}[Proof sketch]
The first few steps are identical to the proof of \Cref{lem:new}: We have a hitting set $S\subset V$ of  size $|S| =\widetilde O(n/D)$ and we compute $d(s,u)$ for all $(s,u)\in S\times V$. We obtain the vertex partition $V = R\cup \bigcup_{i=1}^h H_i$ where $|H_i|=\Theta(d), \diam(H_i)=L=O(1)$ and $\max_{r\in R}\deg(r)\le d$. Here $1\le d<D$ and $h = O(n/d)$.

Then, we 
use \Cref{lem:minplusgroups} 
to compute the $(\min,+)$-product of the distance matrices $d[\bigcup_{i=1}^{h} H_i,S]$ and $d[S,\bigcup_{i=1}^{h} H_i]$.
The time complexity of \Cref{lem:minplusgroups} is (where $q\ge 1$ is a tunable parameter)
   \[ \widetilde O\left (\left (\frac{n}{d}\right )^2\cdot \frac{n}{D} + q\cdot \MM \left (n,\frac{n}{D},n\right )  + \left (\frac{n}{d}\right )^2\cdot \MM\left (d,\frac{n}{Dq},d\right )\right ). \]

We now have $+2$-approximation of $d_D(u,v)$ for all $(u,v)\in \bigcup_{i=1}^{h} H_i \times \bigcup_{i=1}^{h} H_i = (V\setminus R)\times (V\setminus R)$.
Based on these answers, use \Cref{lem:extendtoall} to obtain $+2$-approximation of $d_D(u,v)$ for all $(u,v)\in V\times (V\setminus R)$, in total time $\widetilde O(n^2 d)$, in an analogous way to the last part of the proof of \Cref{lem:new}. 
Then, based on these answers, again use \Cref{lem:extendtoall} to obtain $+2$-approximation of $d_D(u,v)$ for all $(u,v)\in V\times V$ in total time $\widetilde O(n^2 d)$.
\end{proof}

Finally we analyze the overall time complexity obtained from \Cref{lem:newnew}:

\begin{proof}[Proof of \Cref{thm:mainplus2}]
As before, we enumerate $1\le D\le n$ that are powers of two, and compute distance estimates $\tilde d(u,v) \in [d(u,v),d(u,v)+2]$ for pairs $(u,v)$ satisfying $\max_{x\in P(u,v)}\deg(x) \in [D,2D]$ (note that $d_D(u,v)=d(u,v)$ holds for such $u,v$). Finally we combine the answers across all $D$. Then, for given $D$, we can without loss of generality assume the input graph has maximum degree at most $2D$, and hence number of edges at most $O(Dn)$.
We run either \Cref{lem:newnew} or  \Cref{cor:sparse} (whichever is faster) to compute the distance estimates. The time complexity of \Cref{cor:sparse} is $\widetilde O(n^2 \sqrt{D}) $.  The overall time complexity for $+2$-APSP is thus
   \[ \widetilde O\left (\max_{1\le D\le n} \min \left \{n^2\sqrt{D}, 
   \min_{1\le d<D, q\ge 1} \left \{n^2 d + \left (\frac{n}{d}\right )^2\cdot \frac{n}{D} + q\cdot \MM \left (n,\frac{n}{D},n\right )  + \left (\frac{n}{d}\right )^2\cdot \MM\left (d,\frac{n}{Dq},d\right )\right \}
   \right \}\right ).\]

We first show how to set the parameters $d$ and $q$ in terms of $D$ when we only use square matrix multiplication.
 We will show that we can set $d$ and $q$ such that $n\ge Ddq$. If this is the case,
for fixed $D$, the running time becomes
 \[\min \left \{n^2\sqrt{D}, 
   \min_{1\le d<D, q\ge 1} \left \{n^2 d + \left (\frac{n}{d}\right )^2\cdot \frac{n}{D} + q\cdot \frac{n^\omega}{D^{\omega-2}}
 + \frac{n^3}{Dd^{3-\omega}q}\right \}
   \right \}.\]
   
Now, assuming that $n\geq Dd$, set $q = (\frac{n}{Dd})^{\frac{3-\omega}{2}}$ to balance the latter two terms in the running time above. Verify that with this setting of $q$, $\frac{n}{Dd}\ge q$, as needed.

The time complexity is now
\[\min \left \{n^2\sqrt{D}, 
   \min_{1\le d<D} \left \{n^2 d + \frac{n^3}{d^2D} + \Big (\frac{n^{3+\omega}}{D^{\omega-1}d^{3-\omega}}\Big )^{1/2}\right \}
   \right \}.\]



Note that $n^2d\geq \frac{n^3}{d^2D}$ iff $d\ge (n/D)^{1/3}$, and $n^2d\geq \left (\frac{n^{3+\omega}}{D^{\omega-1}d^{3-\omega}}\right )^{1/2}$
iff $d\ge (n/D)^{\frac{\omega -1}{5-\omega}}\geq (n/D)^{1/3}$. Hence to minimize the running time, we set $d= (n/D)^{\frac{\omega -1}{5-\omega}}$ which gives us a running time of 
\[\min \left \{n^2\sqrt{D}, 
   n^2 \cdot (n/D)^{\frac{\omega-1}{5-\omega}}
   \right \}.\]

We verify that this setting of $d$ also gives us that $n\geq Ddq$, as desired.


Finally, the running time is the maximum over all $D\leq n$ of the above quantity.
The worst-case running time is achieved when $n^2\sqrt{D}= n^2 \cdot (n/D)^{\frac{\omega-1}{5-\omega}}$, i.e. when
$D = n^{\frac{2\omega -2}{\omega+3}}$.
Thus, the final runtime for +2 APSP in terms of $\omega$ is $\widetilde{O}(n^{2+\frac{\omega-1}{\omega+3}})$.

For the current value of $\omega$, $\omega<2.371339$ \cite{AlmanDWXXZ25}, the bound is $n^{2.25531}$. As $\omega$ goes to $2$ the runtime goes to $n^{2.2}$, and the exponent $2+\frac{\omega-1}{\omega+3}$ is smaller than $\omega$ if $\omega>\sqrt{5}\approx 2.236$.

To bound the running time using rectangular matrix multiplication we proceed as follows. Let $D=n^b, d=n^\delta, q=n^z$ and minimize
\[\max\{2+b/2,2+\delta,3-2\delta-b,z+\omega(1,1-b,1),2-2\delta+\omega(\delta,1-z-b,\delta)\}.\]

We use the code of \cite{balancer} updated with the newest rectangular matrix multiplication bounds \cite{AlmanDWXXZ25} and obtain that $b=0.45095703,\delta=0.22547851, z=0.15981814$ gives running time $O(n^{2.225479})$.
%
\end{proof}

\section{Faster \texorpdfstring{$+2k$}{+2k}-Approximate APSP}\label{sec:faster+2k}
In this section we use the clustering technique introduced in \Cref{sec:decompositionlemma} to construct a new algorithm for $+2k$-approximate APSP, resulting in a faster runtime for $+4$ and $+6$-approximate APSP. We note that for additive error of $+8$ and up this approach is no longer faster than that of Saha and Ye \cite{sahaYeAPSP}. We prove \Cref{thm:mainplus2k}.

\TheMainPlusTwoK*



Using standard techniques as explained in \Cref{cor:warmup}, we can assume the maximum degree on the true shortest path is $\Theta(D)$. We will use the following generalization of \Cref{lem:new}.

\begin{lemma}\label{lem:newgeneral}
    Let $G=(V,E)$ be an $n$-vertex undirected unweighted graph, $U\subset V$ and parameter $1\leq D \leq n$.
    Denote $d_D(u,v)$ as the minimum length of any (not necessarily simple) path $P$ from $u$ to $v$ such that $\max_{x\in P} \deg(x) \in [D,2D]$.
    We can compute distance estimates $\tilde{d}(u,v)$  such that $d(u,v)\le \tilde{d}(u,v)\leq d_D(u,v)+2$ for all $u\in U,v\in V$, by a deterministic algorithm with time complexity
    \[ \widetilde O\left (\min_{1\le d<D} \left \{|U|\cdot n d + \frac{n}{d} \MM\left (|U|,\frac{n}{D},d\right )\right \}\right ).\]
\end{lemma}

The proof of this lemma follows directly from the proof of \Cref{lem:new}, by restricting the $(\min,+)$-products and graph search to the subset $U$.

Next, assuming we have computed a $+2$-approximation for distances out of a hitting set $U$ which hits the neighborhood of all vertices of degree $ \ge \delta$, we use the following lemma to extend these estimates to a $+2k$-approximation to all distances.

\begin{lemma}\label{lem:generalizetok}
    Let $G=(V,E)$ be an $n$-vertex undirected unweighted graph and parameters $1\leq \delta \le D \leq n$. Let $U \subset V$ be a hitting set of size $|U| = \widetilde O(n/\delta)$ for the neighborhoods of vertices of degree $\ge \delta$. Given distance estimates $\tilde{d}(u,v)$ such that
    $d(u,v)\le \tilde{d}(u,v)\leq d_D(u,v) + 2$ for every pair $u\in U,v\in V$, we can compute distance estimates $\tilde{d}(u,v)$ such that $d(u,v)\le \tilde d(u,v)\le d_D(u,v)+2k$ for all $u,v\in V$, by a deterministic algorithm with time complexity
    \[
    O(n^2\delta^{1/(k-1)}).
    \]
\end{lemma}

\begin{proof}
Without loss of generality, assume $G$ has maximum degree at most $2D$.
    Define the degree thresholds $1=d_0<d_1<\ldots<d_{k-1} = \delta$ by $d_i = \delta^{i/(k-1)}$. For every $1\leq i\leq k-2$, deterministically construct a hitting set $S_i$ of size $|S_i|=\widetilde{O}(n/d_i)$ that hits the neighborhoods of all vertices of degree $\ge  d_i$ (\cref{lem:dethittingset}). Set $S_{k-1}=U$, and $S_0 = V$.
    For every $0\le i\le k-2$, define the edge set $E_i$ to include all edges adjacent to vertices of degree $< d_{i+1}$. For every vertex $v$ of degree $\ge d_{i+1}$, include in $E_i$ an edge connecting $v$ to an arbitrary neighbor in $S_{i+1}$.

Recall that we initially have distance estimates $\tilde d(u,v)$ between $u\in U,v\in V$; initialize $\tilde d(u,v)$ to    $+\infty$ for the remaining pairs.
    Beginning with $i=k-2$ and going down to $i=0$, run the following Dijkstra's searches. For every vertex $u\in S_i$, add an edge from $u$ to every $v\in V$ of weight $\tilde{d}(u,v)$. Run Dijkstra's algorithm out of $u$ on the union of these edges and the edge set $E_i$. Update $\tilde{d}$ with the new distances computed.

    Note that $|E_i|=O(n\cdot \delta^{(i+1)/(k-1)})$ and so running $|S_i|$ Dijkstra's searches (where $0\le i\le k-2$) takes total time \[ \widetilde O(|S_i| \cdot (n + |E_i|))=
    \widetilde{O}\left(\frac{n}{\delta^{i/(k-1)}}\cdot n\delta^{(i+1)/(k-1)}\right) = \widetilde{O}\left(n^2\delta^{1/(k-1)}\right).
    \]
    Assuming $k=O(1)$, the total runtime of running all the searches is $\widetilde{O}(n^2\delta^{1/(k-1)})$. We are left to prove the bound on the approximation error. We do so using the following inductive claim.

    \begin{claim}
        After searching out of $S_i$, $d(u,v)\leq \tilde{d}(u,v)\leq d_D(u,v) + 2(k-i)$ for every $u\in S_i,v\in V$.
    \end{claim}
    \begin{proof}
        All distances computed in the Dijkstra's searches stem from true paths in the graph, so we always have $\tilde{d}(u,v)\geq d(u,v)$. We now show the upper bound inductively.

        For $i=k-1$ we do not perform a search out of $S_{k-1}=U$ but rather use the provided distance estimates that are guaranteed to have $\tilde{d}(u,v)\leq d_D(u,v) + 2 = d_D(u,v) + 2(k-i)$ as required. 

        Now, for $0\le i \le k-2$, assuming the claim holds for $S_{i+1}$, consider a vertex pair $u\in S_i,v\in V$, and the path $P$ from $u$ to $v$ that realizes $d_D(u,v)$ with $\max_{x\in P}\deg(x) \in [D,2D]$. 
        Let $w$ be the closest vertex to $v$ on $P$ with $\deg(w) \geq  d_{i+1}$ and let $s\in S_{i+1}$ be a neighbor of $w$ such that $(s,w)\in E_{i}$. 
        
        Since $\max_{x\in P}\deg(x) \in [D,2D]$, whereas all vertices beyond $w$ on $P$ have degree $<d_{i+1}\leq d_{k-1} = \delta \le  D$, we have that $\max_{x\in P_{uw}}\deg(x)\in [D,2D]$, where $P_{uw}$ denotes the prefix of $P$ up to node $w$. 
        Appending node $s$ to $P_{uw}$ yields a path from $u$ to $s$ which has a highest degree in $[D,2D]$ and has length $|P_{uw}|+1 = d_D(u,v) - d(w,v) + 1$. Hence, $d_D(u,s) \le d_D(u,v)-d(w,v)+1$.  Since $s\in S_{i+1}$, by the inductive hypothesis we conclude $\tilde{d}(u,s) \le d_D(u,s)+2(k-i-1) \le d_D(u,v)-d(w,v)+1 +2(k-i-1)$.

        When running Dijkstra's out of $u \in S_i$ we have an edge from $u$ to $s$ of weight $\tilde{d}(u,s)$. Furthermore, the edge $(w,s)\in E_i$, and the suffix of the path $P$ from $w$ to $v$ is contained in $E_i$. Therefore,
        \[
        \tilde{d}(u,v)\leq \tilde{d}(u,s) + 1 + d(w,v) \leq d_D(u,v)-d(w,v) + 1 + 2(k-i-1) + 1+d(w,v) = d_D(u,v) + 2(k-i),
        \]
        which proves the claim.
    \end{proof}
    Thus, after running Dijkstra's out of $V=S_0$ we have that $\tilde{d}(u,v) \leq d_D(u,v) + 2k$ for every pair of vertices $u,v\in V$. 
\end{proof}

We can now combine \Cref{lem:newgeneral}, \Cref{lem:generalizetok} and \Cref{cor:sparse} to prove \Cref{thm:mainplus2k}.

\begin{proof}[Proof of \Cref{thm:mainplus2k}]
    Enumerate over $1\leq D \leq n$ that are powers of two and compute distance estimates $\tilde{d}(u,v)\in [d(u,v), d(u,v) + 2k]$ for pairs $u,v$ satisfying $\max_{x\in P(u,v)}\deg(x)\in [D,2D]$ (note that $d_D(u,v)=d(u,v)$ holds for such $u,v$). Finally we combine the answers across all values of $D$ by taking the minimum value of $\tilde{d}(u,v)$ computed for every pair, since we are guaranteed that for all values of $D$, the estimate computed is an upper bound to the true distance.

    Thus, for a given $D$, we can assume the given input graph has maximum degree at most $2D$. We can compute our distance estimate in one of the following two ways. First, we can use \Cref{cor:sparse} to compute $\tilde{d}(u,v)$ in time $\tilde{O}(n^2D^{\frac{1}{k+1}})$. 
    
    Otherwise, we can set a parameter $\delta\leq D$ and deterministically construct a hitting set $U$ of size $\tilde{O}\left(\frac{n}{\delta}\right)$ that hits the neighborhood of all vertices of degree $\geq \delta$ via \cref{lem:dethittingset}. Using \Cref{lem:newgeneral}, compute distance estimates $\tilde{d}(u,v)$ such that $d(u,v)\le \tilde{d}(u,v)\leq d_D(u,v)+2$ for every $u\in U,v\in V$.
    Next, use \Cref{lem:generalizetok} to extend these distance estimates to $+2k$-approximation for all pairs $u,v\in V$ 
    in time $\tilde{O}(n^2\delta^{\frac{1}{k-1}})$. 
    
    For every $D$, we can pick the faster of these two options. In total, our runtime comes out to

    \[
    \tilde{O}\left(\max_{1\leq D\leq n} \min \left\{ n^2 D^{\frac{1}{k+1}}, n^2\delta^{\frac{1}{k-1}} +\min_{1\leq d\leq D}\left\{ \frac{n^2 d}{\delta} +\frac{n}{d}\MM\left(\frac{n}{\delta}, \frac{n}{D}, d \right)\right\} \right\} \right).
    \]

    This expression is maximized when the two runtimes are equal, at which point let $0\leq x \leq 1$ be such that $D=\Theta(n^x)$. Set $\delta = n^{\frac{k-1}{k+1}x}$ and $d=n^{\frac{k}{k+1}x}$ to obtain a running time of 
    \[
    \widetilde{O}\left(n^{2+\frac{1}{k+1}x} + n^{1-\frac{k}{k+1}x}\MM\left(n^{1-\frac{k-1}{k+1}x},n^{1-x}, n^{\frac{k}{k+1}x}\right) \right).
    \]

    To minimize this expression, we find the value of $x$ for which \[2+\frac{1}{k+1}x = 1-\frac{k}{k+1}x + \omega\left(1-\frac{k-1}{k+1}x, 1-x, \frac{k}{k+1}x\right).\]  Equivalently $1+x = \omega(1-\frac{k-1}{k+1}x, 1-x, \frac{k}{k+1}x)$.  Thus, we obtain a $+2k$-approximate APSP algorithm running in time $\tilde{O}(n^{2+\frac{1}{k+1}x})$ for the value of $x$ satisfying $1+x = \omega(1-\frac{k-1}{k+1}x, 1-x, \frac{k}{k+1}x)$.
\end{proof}

\section{Open Problems}
\label{sec:open}
A major open question is whether $+2$-approximate APSP can be solved in $n^{2+o(1)}$ time. However this is unknown even for the much easier $2$-multiplicative approximate APSP problem, which can be solved in slightly superquadratic time using fast matrix multiplication techniques \cite{DoryFKNWV24,sahaYeAPSP}. 

The work of Dor, Halperin and Zwick \cite{DHZ00} showed an $\tilde{O}(n^2)$ time algorithm for a $+\log n$ approximate APSP. It remains open to determine if there exists a constant $C$ such that $+C$-approximate APSP can be solved in $n^{2+o(1)}$ time. One can also ask, can we solve $+2$-approximate APSP in $\widetilde O(n^{f(\omega)})$ time, where function $f(\omega)$ satisfies $f(\omega)<\omega$ for all possible values of $\omega>2$? In other words, is $+2$-approximate APSP easier than matrix multiplication?

In the scope of this work, we note that our $+2k$-approximation algorithm follows the idea of the \emph{sparse} approximate APSP algorithm of \cite{DHZ00} in that it computes \emph{all} distances out of each hitting set. In reality, for the next step of the algorithm we are only interested in distances between a hitting set $S_i$ and the next hitting set $S_{i-1}$. The dense approximate APSP algorithm of \cite{DHZ00} makes use of this distinction to speed up their sparse algorithm. The work of Saha and Ye \cite{sahaYeAPSP} also makes use of this fact when computing the bounded-difference $(\min,+)$-product between these hitting sets. However, in order to make use of this fact in our setting we would need to be able to adapt our decomposition lemma (\Cref{lem:decomp}) to decompose a subset of the graph, and not the entire vertex set. It remains open to determine if such a decomposition is possible.


\bibliographystyle{alphaurl} 
\bibliography{main}
\end{document}